\documentclass[11pt]{article}
\usepackage{amsmath}
\usepackage{amsthm}
\usepackage{fullpage}
\usepackage{natbib}
\usepackage{hyperref}

\usepackage{booktabs} % For formal tables
\usepackage[ruled]{algorithm2e} % For algorithms

\SetAlFnt{\small}
\SetAlCapFnt{\small}
\SetAlCapNameFnt{\small}
\SetAlCapHSkip{0pt}
\IncMargin{-\parindent}

\usepackage[textsize=tiny]{todonotes}
\usepackage{amsfonts}

\newcommand{\iccom}[1]{\todo[color=blue!30]{IC: #1}}

\newtheorem{theorem}{Theorem}
\newtheorem{definition}{Definition}
\newtheorem{lemma}[theorem]{Lemma}
\newtheorem{corollary}[theorem]{Corollary}
\newtheorem{remark}{Remark}
\newtheorem{claim}{Claim}
\newtheorem{example}{Example}

\DeclareMathOperator*{\argmin}{arg\,min}
\DeclareMathOperator*{\E}{\mathbb{E}}

\allowdisplaybreaks
\setcitestyle{authoryear}

\title{An impossibility result for strongly group-strategyproof multi-winner approval-based voting}

\author{Ioannis Caragiannis\thanks{Aarhus University, email: \url{iannis@cs.au.dk}} \and Rob LeGrand\thanks{Angelo State University, email: \url{rob.legrand@angelo.edu}} \and Evangelos Markakis\thanks{Athens University of Economics and Business \& Input Output Global, email: \url{markakis@gmail.com}} \and Emmanouil Pountourakis\thanks{Drexel University, email: \url{manolis@drexel.edu}}}

\begin{document}
% Title page for title and abstract only.
\maketitle

% Abstract. Note that this must come before \maketitle.
\begin{abstract}
Multi-winner approval-based voting has received considerable attention recently. A voting rule in this setting takes as input ballots in which each agent approves a subset of the available alternatives and outputs a committee of alternatives of given size $k$. We consider the scenario when a coalition of agents can act strategically and alter their ballots so that the new outcome is strictly better for a coalition member and at least as good for anyone else in the coalition. Voting rules that are robust against this strategic behaviour are called strongly group-strategyproof. We prove that, for $k\in \{1,2, …, m-2\}$, strongly group-strategyproof multi-winner approval-based voting rules which furthermore satisfy the minimum efficiency requirement of unanimity do not exist, where $m$ is the number of available alternatives. Our proof builds a connection to single-winner voting with ranking-based ballots and exploits the infamous Gibbard-Satterthwaite theorem to reach the desired impossibility result. Our result has implications for paradigmatic problems from the area of approximate mechanism design without money and indicates that strongly group-strategyproof mechanisms for minimax approval voting, variants of facility location, and classification can only have an unbounded approximation ratio. 
\end{abstract}

% Paper body
\section{Introduction}
%P1: general about approval voting
{\em Approval voting} offers a simple
and easy-to-use format for running elections on multiple issues
with binary domains. This may involve either committee elections among a set of candidates (which is the viewpoint adopted in this work), or elections for a set of topics that need to be decided upon simultaneously, often referred to as
multiple referenda. Under such a voting format, the voters are allowed to express approval preferences, i.e., each voter can specify an approval set, with as many candidates as she likes, with the interpretation that these are the candidates she is happy with, whereas all other candidates are disapproved by her. In the last decades, many scientific
societies and organizations have adopted approval voting for their council elections, including, among many others, the American Mathematical Society and the Game Theory Society. 
%We refer to \cite{LS23} for a more complete treatment of approval voting elections and their applications.
%and the European Association for Theoretical Computer Science (EATCS).

%P2 Resistance to manipulation
We are interested in the strategic behavior of voters in approval-based elections. Namely, our work is centered around the design of mechanisms for selecting a committee of a given (a priori fixed) size, which falls under the broader umbrella of {\em mechanism design without money}. To begin with, the first goal one would pursue is to obtain {\em strategyproof mechanisms}, where no voter has an incentive to unilaterally misreport her approval preferences to her own benefit. The famous Gibbard-Satterthwaite theorem \cite{Gib73,Sat75} is not applicable in the approval voting domain, as it requires a different space of preferences that express a strict ranking over the set of candidates. Some known extensions of this theorem to weak rankings ---e.g., a result by ~\citet{B93}--- are not applicable to approval voting either. Thus, it is feasible to have strategyproof mechanisms (that are not necessarily dictatorial), as already pointed out by \citet{CKM10}. As an example, for committees of size $k$, the {\em minisum rule}, which selects the candidates with the $k$ highest approval scores, is strategyproof subject to appropriate tie-breaking.

%P3 - resilience to coalitional deviations
Going beyond strategyproofness, several works have also examined coalitional manipulations, where a group of voters may benefit from a joint deviation. This gives rise to two distinct notions of resilience, depending on how we interpret what it means for a coalition to benefit from a joint deviation. In particular, a standard class is that of {\em weakly group-strategyproof} (in short, weakly GSP, or sometimes simply referred to as GSP) mechanisms, where no deviation can make all members of a coalition strictly better off. A stronger definition is that of strongly GSP mechanisms where no deviation makes one member of the coalition better off without worsening the other members.
%P4 Our focus: strongly GSP
Our interest here is in the latter notion of {\em strong group-strategyproofness}. Even though the definitions of weak and strong group-strategyproofness may not seem to differ significantly at first sight, the current literature conveys a different picture. There are positive results on the existence of weakly GSP mechanisms with desirable properties (e.g., satisfying some form of economic efficiency or approximating some social cost objective), both in voting \cite{CKM10}, but also in other fields of mechanism design without money, such as in facility location problems \cite{PT13}. On the contrary, many of these results do not generalize for strongly GSP mechanisms, and in some cases it has remained an open problem whether comparable results are possible with strongly GSP mechanisms. This is precisely the focus of our work.

\subsection{Our Results} 
We consider the problem of selecting a committee of a given fixed size, in elections where voters submit approval ballots. Our main result is in Section \ref{sec:impossibility}, where we show that unanimous, strongly GSP mechanisms do not exist. Unanimity is a very simple and minimal form of efficiency, requiring that if all voters approve the same committee, this should also be the outcome of the election. We view this as a severe impossibility result for strongly GSP mechanisms. Our proof is based on a construction that exploits the Gibbard-Satterthwaite theorem; even though we are interested in approval-based elections, we essentially provide a reduction from appropriately defined elections with ranking preferences, where the non-existence of strategyproof mechanisms implies the impossibility of strongly GSP mechanisms in our setting. We feel that our proof technique is of independent interest and could have further applications.
As a corollary, our result also yields an incompatibility between Pareto-efficiency and strong group-strategyproofness under the approval voting format. 
Finally, we stress that our results are not based on any computational complexity assumption and, hence, hold for exponential-time mechanisms as well.

We proceed in Section \ref{sec:applications}, by demonstrating some further implications of our main result. Most importantly, we resolve one of the open questions posed by \citet{CKM10}, concerning the {\em minimax approval voting} rule, proposed by~\citet{BKS07} as an alternative to the minisum rule. Namely, we prove that there is no strongly GSP mechanism that achieves a finite approximation for the minimax solution. Within the voting domain, we also obtain an impossibility result for {\em participatory budgeting} under approval ballots. Beyond voting, we also have implications for other related problems. The first one is a special case of {\em binary classification}, where we seek a classifier that labels a fixed number of points as positive and the remaining ones as negative. Again, our main result translates to a lack of finite approximations under strongly GSP mechanisms. Finally, the last implication is for {\em constrained facility location}, where the facilities can be placed only in particular locations. 

\subsection{Further Related Work}
For a more complete exposition of approval voting elections and their applications, we refer to the book of~\citet{LS23}. Regarding resistance to manipulation, there are several positive results for strategyproof mechanisms that generalize to weakly GSP mechanisms. In fact, this connection has been studied more thoroughly by \citet{LZ09} and \citet{BBM10}, where sufficient conditions have been identified for domains in which strategyproof mechanisms coincide with weakly GSP mechanisms. Regarding strongly GSP mechanisms in settings without monetary payments, the literature is rather scarce. We are aware of characterization results for the case of two alternatives by \citet{M12} and by \citet{BBM11}. In particular, for approval elections with two alternatives, Manjunath identifies essentially two voting rules that are strongly GSP and Pareto-efficient. On a slightly different direction, the work of \citet{BBM11} establishes characterizations based on certain monotonicity properties. Finally, in terms of impossibility results, \citet{FL+17} show that there are no anonymous, strongly GSP mechanisms for a facility location problem, where the goal is to place a facility on a line, under double-peaked preferences. This is a quite different domain however, and does not have any implications for our voting setting. 

\section{Definitions}
We begin by introducing our notation. We let $N$ denote a set of $n$ voters (or {\em agents}) and $A$ a set of $m$ candidates (or {\em alternatives}).
We consider elections where both the preferences and the ballots are approval-based, i.e., an agent's ballot expresses his approval for each alternative in the subset under consideration.
Hence, a voting profile $P$ is a tuple $P=(P_1, \ldots, P_n)$, where $P_i\subseteq A$ denotes the declared preference of agent $i$.
%Note that we can also represent the preferences of voters as binary vectors in $\{0, 1\}^m$, where a $1$ in the $i$-th coordinate indicates approval of the $i$-th alternative. 
%We will mostly stick to the set representation, but when convenient, we will switch to the binary vector representation.\vmcom{probably we do not need this} 
%(as in Section~\ref{sec:manip}). 
Under this notation, an {\em approval election} is specified by a tuple $(N, A, P)$. Our focus is on multi-winner elections for selecting a committee of some predetermined size. We let $k$ denote the committee size. So, for any voting rule in our setting, the outcome must be some set $S\subseteq A$, with $|S|=k$. To distinguish with the rules discussed later in Section~\ref{sec:rankings}, we refer to them as {\em multi-winner approval-based voting rules}. 

%There are several voting rules that can be defined under the approval-based format, for committee elections; see \cite{Kilgour10} for a review. Arguably the most commonly used method is candidate-wise majority, referred to as the {\em minisum} rule. When there is no cardinality constraint on the size of the committee, minisum selects the candidates that are approved by a majority of voters. When the committee needs to be of a fixed size $k$, then minisum simply selects the $k$ most approved candidates. 

We extend the notion of {\em Hamming distance} to subsets of $A$ as follows. We say that the distance between two sets
$Q$ and $T$ is the total number of alternatives in which they differ, i.e.,
\[d(Q,T) = |Q\setminus T|+|T\setminus Q|=|Q|+|T|-2|Q\cap T|.\]
Note that this is precisely the Hamming distance of the sets when represented as binary vectors, where the $i$th coordinate of each vector equals $1$ if the $i$th alternative belongs to the set and equals $0$ otherwise.

\subsection{Resistance to manipulation and other properties}
We will assume that agents are strategic and may misreport their preferences to the voting rule if this can be in their interest. We assume that each agent prefers committees that include as many alternatives from their approval preference as possible. Equivalently, agents evaluate a committee in terms of its Hamming distance to their approval preference (with the smaller Hamming distance, the better).

Hence, it is desirable for voting rules to be resistant to possible manipulations by the agents. Given a profile $P$ and a voting rule $R$, we denote by
$R(P)$ the outcome of the voting rule on profile $P$. We also denote by $P_{-i}$ the preferences of all agents besides $i$.
Hence, we can also write $P$ as $(P_i, P_{-i})$. Our first property indicating resistance to manipulation is {\em strategyproofness}, which requires that no agent $i$ has an incentive to unilaterally change her preference so as to reduce the distance of $P_i$ from the outcome of the voting rule.
\begin{definition}
\label{def:SP}
	A voting rule $R$ is {\em strategyproof} (SP) if for any election $(N,A,P)$, any agent $i\in N$, and any set of alternatives $P_i'\subseteq A$, it holds that
	$$ d(P_i, R(P_i, P_{-i})) \leq d(P_i, R(P_i', P_{-i})).$$
\end{definition}

Moving on, we can also consider deviations by coalitions of agents and define two analogous notions of resistance to manipulation. In particular, the first such notion says that there is no deviation that can make all deviating agents strictly better off. Given a profile $P$, in analogy to $P_i$ and $P_{-i}$, we let $P_S$ and $P_{-S}$ denote the preferences of a set of agents $S$ and 
 of all agents except $S$, respectively.

\begin{definition}
	A voting rule $R$ is {\em weakly group-strategyproof} (weakly GSP) if for any election $(N,A,P)$ and any coalition $S\subseteq N$ of agents, there is no profile $P_S'$ of the agents in $S$ such that
	$$ d(P_i, R(P_S, P_{-S})) > d(P_i, R(P_S', P_{-S}))$$
 for every agent $i\in S$.
\end{definition}

The notion that is the main focus of our work is a stronger requirement, defined as follows.
\begin{definition}
	A voting rule $R$ is {\em strongly group-strategyproof} (strongly GSP) if for any election $(N,A,P)$ and any coalition $S\subseteq N$ of agents, there is no profile $P_S'$ of the agents in $S$ such that
	$$ d(P_i, R(P_S, P_{-S})) \geq d(P_i, R(P_S', P_{-S}))$$ 
 for every agent $i\in S$, with strict inequality for at least one agent of $S$.
\end{definition}

The rationale behind this last concept is that we demand the voting rule to be resistant to coalitions in which some of the agents may change their preference profile in order to help other members of the coalition (without necessarily gaining something for themselves). Clearly, strong group-strategyproofness is a stronger notion than weak group-strategyproofness. 

Apart from resistance to manipulation, we will examine two properties related to the efficiency of voting rules. The first is a very natural axiom for our setting. 
\begin{definition}
    We say that a voting rule $R$ for selecting a committee of size $k$ is {\em unanimous} if whenever all agents approve the same set $S$ of $k$ alternatives (and nothing else), $R$ outputs $S$ as the selected committee.
\end{definition}

\begin{definition}
    A voting rule $R$ is {\em Pareto-efficient} if, for every input profile $P$, it outputs a set $S$ of $k$ alternatives, such that there is no other committee $S'$ of size $k$, with at least one agent being closer to $S'$ than $S$, and all other agents not being further off. 
\end{definition}
%\vmcom{keep definition of PO?}

\subsection{Voting by ranking ballots}
\label{sec:rankings}

Although our main focus is on approval voting, the proof of our main result is based on a construction that involves elections
where the ballot of each agent is a strict ranking of the alternatives in decreasing order of preference. Hence, a preference profile in ranking-based voting is given by a tuple $(\succ_1, \cdots, \succ_n)$, where $\succ_i$ is the ranking of agent $i$ on the set of alternatives $A$. For brevity, we use $\succ$ to denote the entire preference profile $(\succ_1, \cdots, \succ_n)$. Furthermore, for two alternatives $p, q\in A$, we will use $p \succ_i q$ to denote that agent $i$ prefers $p$ to $q$ (and, hence, $p$ appears higher than $q$ in the ranking of $i$).

Given a set of agents $N$, a set of alternatives $A$, and a preference profile $\succ$, a {\em ranking-based election} is specified by a tuple $(N, A, \succ)$. 
We will consider {\em single-winner ranking-based voting rules}, which return a single alternative as an outcome when applied to a ranking profile. Hence, for a voting rule $T$, we let $T(\succ)$ denote the winning alternative of $T$ when given the preference profile $\succ$ as input.

In analogy to Definition~\ref{def:SP}, we can also define the notion of strategyproofness in this setting as well. In the definition below, for a profile $\succ$, we denote by $\succ_{-i}$ the preference profile of all agents except $i$.

\begin{definition}
\label{def:SP-rankings} 
A ranking-based voting rule $T$ is strategyproof (SP) if for any ranking-based election $(N,A,\succ)$, any agent $i\in N$, and any ranking $\succ_i'$ of the alternatives in $A$, it holds that
	$$ T(\succ) \succ_i T(\succ_i', \succ_{-i}).$$
\end{definition}

The classic impossibility result due to Gibbard and Satterthwaite, states that any SP and {\em onto} single-winner ranking-based voting rule for elections with at least three alternatives is a dictatorship. Here, a voting rule $T$ is onto if, for every alternative $p\in A$, there exists a profile $\succ$ such that $T(\succ) = p$. A voting rule $T$ is a dictatorship if there exists an agent $i$ such that for every profile $\succ$, the outcome $T(\succ)$ is the top choice of agent $i$.

\begin{theorem}[\citealt{Gib73,Sat75}]
\label{thm:GS}
    Any single-winner ranking-based voting rule for ranking-based elections with at least three alternatives, which is SP and onto, must be dictatorial.
\end{theorem}

\section{Our Impossibility Result}
\label{sec:impossibility}

Before we embark on our impossibility result, it is instructive to start with a discussion on the existence of SP and weakly GSP voting rules.
If we care only about strategyproofness, the minisum voting rule (with appropriate deterministic tie-breaking) is SP. When it comes to coalitional manipulations, it would be expected that some form of dictatorial mechanisms would be weakly, or even strongly GSP. One crucial issue here with defining a dictatorial rule is that the outcome cannot always coincide exactly with the dictator's preferences, due to the constraint on the size of the committee being exactly $k$. Therefore, a voting rule is dictatorial in our setting if it falls within the following class of voting rules (referred to as $k$-completion by \citet{LMM07}).

\begin{quote}
    \noindent {\bf Dictatorial $k$-completion:} Pick an agent $i$ as the dictator. If $|P_i| < k$, then output the union of $P_i$ together with $k-|P_i|$ alternatives from $A\setminus P_i$ (selected according to some tie-breaking order); else if
$|P_i| > k$, then output a subset of $P_i$ of size $k$ (again, according to some tie-breaking order); else return $P_i$.
\end{quote}

\noindent The $k$-completion rules do indeed satisfy some form of resistance to coalitional manipulation, but they are still not strongly GSP.

\begin{theorem}[\citealt{CKM10}]
    Any $k$-completion voting rule is weakly GSP but not strongly GSP.
\end{theorem}

The reason a $k$-completion rule is not always strongly GSP is that, dependent on the tie-breaking rule used, the dictator could help some other agent by changing his approval ballot and steer the rule to select a committee that makes some agents better off, while the dictator is not worse off. 

A natural question, posed already a while ago by \citet{CKM10}, is whether non-trivial strongly GSP voting rules exist. Our main result is the following theorem, where we obtain a negative answer, as long as we demand the additional and seemingly harmless property of unanimity.

\begin{theorem}
\label{thm:main}
    Consider approval elections with $m$ alternatives, where the outcome must always be a committee of size $k$, with $k\in [m-2]$. Then, there is no multi-winner approval-based voting rule that is both strongly GSP and unanimous. 
\end{theorem}

%For example, we observe that the refined $k$-completion algorithm of Section \ref{sec:weak-gsp} is not strongly GSP. If the dictator is indifferent between outcomes, she can restrict her approval sets in order to form coalitions with voters not favored by the ordering. As we will show, strongly GSP algorithms cannot guarantee any approximation ratio for the minimax objective. To do this, we start by drawing a connection between the approximation ratio of strongly GSP algorithms and the properties of unanimity and Pareto-efficiency.

The remaining section is devoted to the proof of this result. We start with the following lemma, which connects unanimity, strongly group-strategyproofness, and Pareto-efficiency.

%\begin{lemma}\label{lem:finite-approx}
%	Any strongly GSP algorithm for $k$-MAV with a finite approximation ratio is Pareto-efficient. 
%\end{lemma}

\begin{lemma}\label{lem:finite-approx}
	Any strongly GSP and unanimous multi-winner approval-based voting rule is Pareto-efficient. 
\end{lemma}

\begin{proof}
Assume that the voting rule is not Pareto-efficient. Then, there exists an approval election $(N,A,P)$, for which the voting rule returns a $k$-sized committee of alternatives $K$, while there exists another $k$-sized committee $K'$ such that $d(K',P_{i^*}) < d(K,P_{i^*})$ for some agent $i^*\in N$ and $d(K',P_i) \leq d(K,P_i)$ for any other agent $i\in N$. But then, all agents have an incentive to form a coalition and misreport their preferences by simply reporting the set of alternatives $K'$. In that case, by the unanimity property, $K'$ should be selected as the outcome. The agent $i^*$ would then become strictly better off, and the rest of the agents would not be worse off. This would contradict the fact that our voting rule is strongly GSP.
\end{proof}

Using Lemma \ref{lem:finite-approx}, we will prove Theorem~\ref{thm:main} by showing that no Pareto-efficient multi-winner approval-based voting rule is strongly GSP. 
We first provide a high-level outline of our proof. We believe that the technique 
%(and its by-product stated in Theorem \ref{thm:gsppareto}) 
can be of independent interest, potentially useful for establishing other impossibility results as well. 
The backbone of the proof is a reduction that transforms a ranking-based election (as defined in Section \ref{sec:rankings}) into an approval election, so that any Pareto-efficient multi-winner approval-based voting rule for the latter is mapped naturally to a single-winner ranking-based voting rule for the former. We show that if the multi-winner approval-based voting rule is furthermore strongly GSP, then the single-winner ranking-based voting rule is SP and onto. Then, the Gibbard-Satterthwaite theorem (Theorem~\ref{thm:GS}) implies certain ``dictatorship-like'' properties for the multi-winner voting rule, eventually leading to a contradiction.

%Following the notation of Section \ref{sec:rankings}, a preference profile in voting by ranking ballots is given by a tuple $(\succ_1, \cdots, \succ_n)$, where $\succ_i$ is a strict linear order of voter $i$ on the set of candidates. We denote such a tuple by $\succ$.

\begin{definition}
	\label{def:reduction}[Construction of an approval election from a ranking-based election] 
	Given a parameter $k$ and a ranking-based election $I=(N,A,\succ)$ with $|N| = n$ and $|A| = m$, we construct a corresponding approval election $=(N',A',P(\succ))$, as follows: 
	\begin{itemize}
		\item For every agent $i\in N$, we introduce $m-1$ different copies of agents, indexed as $(i,1),\dots,(i,m-1)$. Hence, the approval-based election has a total of $(m-1)\cdot n$ agents.
		\item The set of alternatives is $A'=A\cup D$, where $D$ consists of $k-1$ dummy alternatives. Hence, the approval election has a total of $|A'| = m+k-1$ alternatives.
		\item For $i\in [n]$ and $j\in [m-1]$, agent $(i,j)$ has the approval ballot $P_{(i,j)}$ defined as   
		$$P_{(i,j)} = \{\text{top $j$ ranked alternatives in $\succ_i$}\}\cup D$$
		Hence, the approval profile is given by $P(\succ) = \{P_{(i, j)}\}_{i\in [n], j\in [m-1]}$.
	\end{itemize}
\end{definition}

\begin{example}
Before we proceed, we demonstrate the construction with a simple example, when $n$ and $m$ are small. Consider the ranking-based election $(N,A,\succ)$ with $N=\{1,2\}$, $A=\{x, y, z\}$, and with $\succ$ defined as $x\succ_1 y \succ_1 z$ and $y\succ_2 z \succ_2 x$. 
The reduction of Definition \ref{def:reduction} for $k=2$ gives an approval election with four agents $N'=\{(1,1), (1,2), (2,1), (2,2)\}$, and four candidates $A'=\{x, y, z, d_1\}$. Here, we have only one dummy candidate. The approval ballots are: $P_{(1,1)}=\{x,d_1\}$, $P_{(1, 2)}=\{x,y,d_1\}$, $P_{(2, 1)}=\{y,d_1\}$, and $P_{(2,2)}=\{y,z,d_1\}$.
\end{example}

Now, let $(N,A,\succ)$ be a ranking-based election and let $(N',A',P(\succ))$ be the corresponding approval election defined by our reduction in Definition~\ref{def:reduction}. Also, let $R$ be a Pareto-efficient multi-winner approval-based voting rule which outputs the committee $R(P(\succ))$ of size $k$ when applied on the profile $P(\succ)$. We will show that $R(P(\succ))\setminus D$ is a singleton and will use it to define a corresponding single-winner ranking-based voting rule for the ranking-based election.

%The way we defined the solution in \eqref{eq:reduction}, it appears that $T(\succ)$ may not produce a single winner. However, in the next lemma we show that if we start with a Pareto-efficient algorithm for the approval voting setting, then $T(\succ)$ produces a single winner in the voting by rankings setting, i.e., all $k-1$ dummy candidates should be included in any Pareto-efficient algorithm.

\begin{lemma}\label{lem:feasible}
If the multi-winner approval-based voting rule $R$ is Pareto-efficient, then $R(P(\succ))\setminus D$ is a singleton. 
\end{lemma}

\begin{proof}
For the sake of contradiction, assume that $|R(P(\succ))\setminus D|\geq 2$. Note that for every agent $i\in N$, agent $(i,1)$ of $N'$ approves exactly one alternative outside $D$. Hence, there is at least one alternative in the outcome $R(P(\succ))$, say $a^*$, that this agent does not approve. Moreover, each alternative in $D\setminus R(P(\succ))$ is approved by all agents. Now, consider replacing $a^*$ in the outcome of $R$ by some alternative in $D\setminus R(P(\succ))$. Agent $(i, 1)$ is strictly better off now, and all other agents are not worse off. Hence, $R$ is not Pareto-efficient, a contradiction.
\end{proof}

We now define the single-winner ranking-based voting rule $T$, which returns the alternative contained in $R(P(\succ))\setminus D$ on input the profile of rankings $\succ$. We will show that if the multi-winner approval-based voting rule $R$ is both Pareto-efficient and strongly GSP, the rule $T$ satisfies the premises of the Gibbard-Satterthwaite theorem (Theorem \ref{thm:GS}).

\begin{lemma}\label{lem:sp+onto}
If $R$ is a strongly GSP and Pareto-efficient multi-winner approval-based voting rule, the induced single-winner ranking-based voting rule $T$ is onto and SP.
\end{lemma}

\begin{proof}
We first prove that $T$ is onto, i.e., for every alternative $a\in A$, there exists $\succ$ such that $T(\succ)=a$. We construct $\succ$ such that $a$ is the top alternative in every agent's ranking. Note that, in the corresponding approval voting profile $P(\succ)$ defined by our reduction in Definition~\ref{def:reduction}, the only Pareto-efficient committee of size $k$ is $D\cup \{a\}$. Hence, this must be the outcome of $R$, and this means that $T$ is equal to the single alternative in $R(\succ)\setminus D$, i.e., alternative $a$.

Second, we show that $T$ is SP. Assume towards a contradiction that $T$ is not SP, i.e., there exists a preference profile $\succ$, some agent $i$, and some ranking of alternatives $\succ'_i$ such that
\begin{equation}\label{eq23}
T(\succ'_i,\succ_{-i})\succ_i T(\succ_i,\succ_{-i}).
\end{equation}
	
By our reduction, the approval ballots of the profiles $P(\succ_i,\succ_{-i})$ and $P(\succ'_i,\succ_{-i})$ differ only in the agents $(i,1),\dots,(i,m-1)$. Hence, the main idea is to show that the unilateral deviation to $\succ'_i$ by agent $i$ in the ranking-based election corresponds to successful deviations by the coalition of agents $(i,1),\dots,(i,m-1)$ in the approval-based election, from profile $P(\succ)$ to profile $P(\succ'_i,\succ_{-i})$.

Let $a=T(\succ_i,\succ_{-i})$ and $b = T(\succ'_i,\succ_{-i})$. Suppose that according to $\succ_i$, alternative $a$ is in position $r$ and $b$ is in position $r'$. Observe that Equation~\eqref{eq23} implies $r'<r$. %Consider the group deviation  of voters $(i,1),\dots,(i,m)$ in the approval voting setting, where the misreport of their preferences is derived by Definition \ref{def:reduction}, when changing $\succ$ to $(\succ_i', \succ_{-i})$. 
Since we assumed that, under rule $T$, $b$ is the winner in profile $(\succ'_i,\succ_{-i})$ and $a$ is the winner in profile $\succ$, this means that the outcome of rule $R$ in the approval voting election is $D\cup \{a\}$ for profile $P(\succ)$ and $D\cup \{b\}$ for profile $P(\succ'_i,\succ_{-i})$. Note that the agents $(i,j)$ with $r'\leq j <r$ approve alternative $b$ but do not approve alternative $a$. Hence, their distance from the outcome of $R$ decreases in profile $P(\succ'_i,\succ_{-i})$ compared to profile $P(\succ)$. The remaining agents in the coalition are indifferent, since either they approve both $a$ and $b$ (this holds for agents $(i, j)$ with $j\geq r$, if any), or they approve neither of them (this holds for agents $(i,j)$ with $j<r'$, if any). Hence, we have constructed a deviation, where at least one member of the coalition is better off, and some are indifferent. This contradicts the assumption that the multi-winner approval-based voting rule $R$ is strongly GSP. Therefore, the single-winner ranking-based voting rule $T$ is SP.
\end{proof}

Notice that, according to our reduction in Definition~\ref{def:reduction}, if the number of alternatives in the approval election is at least $k+2$, then the number of alternatives in the ranking-based election is at least $3$. Thus, by Lemmas~\ref{lem:feasible},~\ref{lem:sp+onto}, and Theorem~\ref{thm:GS}, we get the following corollary. 
\begin{corollary}\label{cor:its-a-dictatorship}
Let $R$ be a Pareto-efficient and strongly GSP multi-winner approval-based voting rule for approval elections with $m\geq 3$ alternatives and $k\in [m-2]$. Then, the induced single-winner ranking-based voting rule $T$ is a dictatorship, i.e., there exists an agent $i^*$ such that for every profile $\succ$, $T(\succ)$ is the top preference of agent $i^*$.
\end{corollary}

The fact that $T$ is a dictatorship implies that $R$ also has some dictatorship-like attributes. We show that this contradicts the fact that $R$ is strongly GSP, which implies that the intersection of strongly GSP and Pareto-efficient multi-winner approval-based voting rules is empty, completing the proof of Theorem~\ref{thm:main}.

%\begin{theorem}\label{thm:gsppareto}
%For any $k\geq 1$, a strongly GSP voting rule in the approval voting setting for selecting a committee of $k$ candidates out of $m$ candidates, with $m\geq k+2$, cannot be Pareto-efficient.
%\end{theorem}

Consider a unanimous (and, by Lemma~\ref{lem:finite-approx}, Pareto-efficient) and strongly GSP multi-winner approval-based voting rule $R$ which outputs committees of size $k\in [m'-2]$ for elections with $m'\geq 3$ alternatives. Consider the reduction of Definition \ref{def:reduction}, with parameters $k$, $n=3$, and $m=m'-k+1$, with $m\geq 3$, and let $T$ be the single-winner ranking-based voting rule induced by $R$. By Corollary~\ref{cor:its-a-dictatorship}, $T$ is a dictatorship, and without loss of generality, let us assume that the dictator is agent $1$.

Now, let $x$ and $y$ be two specific alternatives of $A$ and consider the rankings $\succ^x$ and $\succ^y$ of the alternatives of $A$, which differ only in the two top positions, defined as follows. Ranking $\succ^x$ has alternative $x$ first, alternative $y$ second, and then the alternatives of $A\setminus\{x,y\}$ in some arbitrary order. Ranking $\succ^y$ has alternative $y$ first, alternative $x$ second, and then the alternatives of of $A\setminus\{x,y\}$ in the same order with $\succ^x$. Based on these, we will construct three distinct approval profiles and argue about their outcome under $R$, that will eventually lead to a contradiction.

Consider first the ranking profiles $(\succ^x,\succ^x,\succ^y)$ and $(\succ^y,\succ^x,\succ^y)$. The corresponding approval profiles $P_x=P(\succ^x,\succ^x,\succ^y)$ and $P_y=P(\succ^y,\succ^x,\succ^y)$ differ only in the ballot of agent $(1,1)$, which is $\{x\}\cup D$ in $P_x$ and $\{y\}\cup D$ in $P_y$. Among the remaining agents in profiles $P_x$ and $P_y$, the ones that are of interest for the arguments below are agent $(2,1)$ with approval ballot $\{x\}\cup D$ and agent $(3,1)$ with approval ballot $\{y\}\cup D$. By the definition of our reduction, the fact that $T$ is a dictatorship of agent $1$, and the relation between the voting rules $R$ and $T$, we have that $R(P_x)=\{x\}\cup D$ and $R(P_y)=\{y\}\cup D$.

Now, consider a third approval profile $P_{xy}$, differing from $P_x$ and $P_y$ only in the ballot of agent $(1,1)$ which is now $\{x,y\}\cup D$. For convenience, the three profiles are depicted in Table~\ref{tab:example}, when $m=3$. We remark that this profile is not produced by our reduction. Still, the Pareto-efficiency of voting rule $R$, implies that $D\subset R(P_{xy})$, since the alternatives in $D$ appear in the approval ballot of all agents in $P_{xy}$. We distinguish between two cases regarding the single alternative in $R(P_{xy})\setminus D$.

First, if $R(P_{xy})\setminus D\not=\{x\}$, then the agents $(1,1)$ and $(2,1)$, with ballots $\{x,y\}\cup D$ and $\{x\}\cup D$ under $P_{xy}$, have Hamming distance at least $1$ and exactly $2$  respectively, from the $k$-sized committee $R(P_{xy})$. The deviation of these two agents to approval ballots $\{x\}\cup D$ for both yields the approval profile $P_x$, for which rule $R$ outputs the committee $\{x\}\cup D$. Its Hamming distance from the approval ballots of agents $(1,1)$ and $(2,1)$ in profile $P_{xy}$ is only $1$ and $0$, respectively, implying the existence of a successful deviating coalition that contradicts the assumption that voting rule $R$ is strongly GSP.

It remains to consider the case $R(P_{xy})=\{x\}\cup D$. Then, the agents $(1,1)$ and $(3,1)$, with ballots $\{x,y\}\cup D$ and $\{y\}\cup D$ under $P_{xy}$, have Hamming distance $1$ and $2$ from the committee $R(P_{xy})$ respectively. Their deviation to ballot $\{y\}\cup D$ for both yields the approval profile $P_y$, for which the rule $R$ outputs the $k$-sized committee $\{y\}\cup D$. Hence, this will yield a Hamming distance of $1$ and $0$ from the approval ballots of agents $(1,1)$ and $(3,1)$ in the approval profile $P_{xy}$. This again implies the existence of a successful deviating coalition, contradicting the assumption that $R$ is strongly GSP.

We conclude that $R$ cannot be Pareto-efficient (and, hence, unanimous) and strongly GSP, completing the proof of Theorem~\ref{thm:main}.

\begin{table}[ht]
\centering
    \caption{An example of the profiles used in the final step of the proof of Theorem~\ref{thm:main}, when the ranking-based election we start from has $m=3$ alternatives. The profiles differ only in the approval ballot of agent $(1,1)$ and any Pareto-efficient outcome of the rule $R$ (i.e., either $\{x\}\cup D$ or $\{y\}\cup D$) on profile $P_{xy}$ violates strong group-strategyproofness. }
    \label{tab:example}
    \begin{tabular}{c|c c c}\toprule
    \text{agent} & $P_x$ & $P_y$ & $P_{xy}$\\\hline
    $(1,1)$ & $\{x\} \cup D$ & $\{y\} \cup D$ & $\{x,y\} \cup D$\\
    $(1,2)$ & $\{x,y\} \cup D$ & $\{x,y\} \cup D$ & $\{x,y\} \cup D$\\
    $(2,1)$ & $\{x\} \cup D$ & $\{x\} \cup D$ & $\{x\} \cup D$\\
    $(2,2)$ & $\{x,y\} \cup D$ & $\{x,y\} \cup D$ & $\{x,y\} \cup D$\\
    $(3,1)$ & $\{y\} \cup D$ & $\{y\} \cup D$ & $\{y\} \cup D$\\
    $(3,2)$ & $\{x,y\} \cup D$ & $\{x,y\} \cup D$ & $\{x,y\} \cup D$\\\hline
    \text{outcome} & $\{x\} \cup D$ & $\{y\} \cup D$ & ?\\\bottomrule
    \end{tabular}
\end{table}

\section{Implications}
\label{sec:applications}
We now discuss some implications of our impossibility result in other settings. We begin by answering a question by~\citet{CKM10} about the ``approximability'' of the minimax approval voting rule. Then, we discuss three problems, which are just generalizations of multi-winner approval-based voting (even though some may seem different at first glance).

\subsection{Approximating minimax approval voting}
\label{sec:minimax}
Minimax approval voting has been proposed by~\citet{BKS07}, who argued about its use in comparison to the minisum rule that is mostly adopted in practice. In particular, the $k$-minimax approval voting rule takes as input a profile of approval ballots and returns a $k$-sized committee of alternatives that minimizes the maximum Hamming distance from the ballots. Unfortunately, as proved by~\citet{LMM07}, the rule has two important drawbacks. The first is that computing the winning committee is an NP-hard problem, and the second is that the rule is not SP.

As a relaxation, \citet{LMM07},~\citet{CKM10}, and~\citet{BS14} studied {\em approximate mechanisms} for $k$-minimax approval voting. Given an approval election $(N,A,P)$ and a $k$-sized committee $C$, let $D(C,P)=\max_{i\in N}{d(C,P_i)}$ denote the maximum Hamming distance of $C$ from the approval ballots of $P$. Also, let $C^*\in \argmin_{K}{D(K,P)}$ be a $k$-sized committee that would be returned by the $k$-minimax approval voting rule on input the election $(N,A,P)$. A voting rule $R$ approximates the $k$-minimax solution within a factor of $\rho\geq 1$, if $D(R(P),P)\leq \rho\cdot D(C^*,P)$ for every profile $P$. The quantity $\rho$ is called the approximation ratio of mechanism $R$.

The papers by~\citet{LMM07} and~\citet{CKM10} present SP and weakly GSP voting rules, which approximate $k$-minimax approval voting within a factor of $3-\frac{2}{k+1}$. Among them, the minisum rule with appropriate tie-breaking is SP, whereas the class of $k$-completion voting rules, presented in Section \ref{sec:impossibility}, is weakly GSP. The question of whether such an approximation factor can be achieved by a strongly GSP mechanism was left open. 

To resolve this open question using our impossibility result, we make first the following claim. 
\begin{claim}
\label{cl:non-unan}
A non-unanimous voting rule has an infinite approximation ratio for $k$-minimax approval voting.
\end{claim}
\begin{proof}
Consider such a non-unanimous rule. Non-unanimity implies that there is a profile $P$ in which all agents approve the same $k$-sized committee of alternatives, say $C^*$, but the rule returns some different committee, say $C$. Clearly, $D(C^*,P)=0$, i.e., the minimax solution has a Hamming distance of $0$ to all voters, whereas $D(C,P)>0$ in this case. This means that the rule has an infinite approximation ratio. 
\end{proof}

The following statement now follows due to Theorem~\ref{thm:main}.
\begin{theorem}
Any strongly GSP approval-based multi-winner voting rule in elections with $m\geq 3$ alternatives, $k\in [m-2]$, and $n\geq 6$ agents has an infinite approximation ratio for $k$-minimax approval voting.
\end{theorem}

\subsection{Participatory budgeting with an approval voting format}
An important application of approval-based voting is in {\em participatory budgeting}~\cite{AS21,RM23}. In the most typical scenario, a municipality has a fixed budget and considers implementing several projects. Each project has a known cost, but the total cost of all projects exceeds the available budget. Thus, only a subset of projects can be selected for implementation. Participatory budgeting is used to delegate the selection of projects to the citizens. Each participating citizen casts a ballot with their preferences, and the municipality has to select, using a predefined mechanism, a subset of projects that are within the budget and reflect the citizens' preferences. Among other formats that have been used for eliciting the citizens' preferences, approval ballots seem to be very popular, not only within municipalities, but also in programs offered by blockchain communities or decentralized autonomous organizations (DAOs); e.g., see the recent related survey by~\citet{T23}.

Let us focus now on the case of projects with equal costs, so that for some integer $k$, all subsets of $k$ projects are the maximal sets of projects that are within the budget\footnote{In fact, for our purposes, the project costs do not need to be the same, as long as any set of $k$ projects is budget feasible, and any superset is infeasible.}. This setting can be easily seen to be equivalent to multi-winner approval voting. In particular, we say that a mechanism is maximal if it always returns a maximally feasible set of projects for funding (the addition of any other project violates the budget constraint). For the case we are considering, a maximal mechanism would have to select a set of exactly $k$ projects (which can be viewed as a committee of size $k$ in the approval voting election).  
Assuming that each citizen aims to have as many of her favourite projects selected, our impossibility result implies that unanimity and strong group-strategyproofness are not compatible for participatory budgeting mechanisms that are maximal.

\begin{theorem}\label{thm:pb}
Unanimous and strongly GSP maximal mechanisms for participatory budgeting with an approval ballot format  do not exist.
\end{theorem}

We remark that Theorem~\ref{thm:pb} applies to mechanisms which allow citizens to cast ballots with arbitrarily many projects, even if their total cost exceeds the available budget. This is not an unusual practice today.\footnote{The website \url{https://en.wikipedia.org/wiki/List_of_participatory_budgeting_votes} has information about the format of participatory budgeting votes held in the last ten years.} Whether Theorem~\ref{thm:pb} carries over for inputs, where each agent's ballot respects the budget restriction is an interesting open problem.

\subsection{Classification with shared inputs}
An important problem in {\em supervised learning} is, given a set of training examples that maps points from an input space to labels, to select a classifier from a concept class that is as close to the data in the training example as possible. \citet{MPR12} study strategic issues in {\em binary classification} by considering the following multi-agent setting with an input space $\mathcal{X}$, a concept class $\mathcal{H}$ and $n$ agents. The concept class contains classifiers; each classifier $h\in \mathcal{H}$ is a function mapping each point of $\mathcal{X}$ to a label from $\{+,-\}$. The training set consists of a set of points $X\subseteq\mathcal{X}$ and data provided by $n$ agents, where agent $i$ provides a labeling $Y_i:X\rightarrow \{+,-\}$ of the points in $X$.

A classifier is evaluated by each agent $i\in [n]$ using the agent's loss function $\ell_i$, which returns the number of data points for which the labeling of the classifier and the labeling provided by agent $i$ differ. i.e., for the classifier $h\in \mathcal{H}$ and agent $i$,
\begin{align*}
    \ell_i(h,Y_i) &= \sum_{x\in X}{\mathbb{I}\{h(x)\not=Y_i(x)\}}.
\end{align*}
Furthermore, each agent $i\in [n]$ has a positive weight $w_i$, so that $\sum_{i\in [n]}{w_i}=1$, indicating the importance of the agent for the classification task. A classification mechanism takes as input the dataset $Y=(Y_1, Y_2, ..., Y_n)$, which consists of the labelings provided by the agents, along with their weight vector $\mathbf{w}=(w_1, w_2, ..., w_n)$, and returns a classifier from the concept class $\mathcal{H}$. The outcome of the classification mechanism, say $h\in \mathcal{H}$, is evaluated using the global risk $L(h,Y,\mathbf{w})$, defined as
%For dataset $Y=(Y_1, ..., Y_n)$ and agent weight vector $\mathbf{w}=(w_1, ..., w_n)$, the global risk $L(h,Y,\mathbf{w})$ of classifier $h\in \mathcal{H}$ is defined as
\begin{align*}
    L(h,Y,\mathbf{w}) &= \sum_{i\in [n]}{w_i\cdot \ell_i(h,Y_i)}.
\end{align*}
The well-known {\em empirical risk minimization} (ERM) algorithm selects a classifier minimizing the global risk, i.e., 
\begin{align*}
    \text{erm}(Y,\mathbf{w}) &\in \argmin_{h\in \mathcal{H}}{L(h,Y,\mathbf{w})}.
\end{align*}

\citet{MPR12} assume that agents are strategic and may decide to misreport their private data to the classification mechanism if the outcome of the mechanism incurs a smaller loss to them. Adapting the notions of strategyproofness and strong group-strategyproofness to this setting, we say that a classification mechanism $M$ taking as input the weight vector $\mathbf{w}$ and the private labelings of the agents, $Y=(Y_1, ..., Y_n)$, for the set of points $X\subseteq \mathcal{X}$, is SP if there is no agent $i\in [n]$ and labeling $Y'_i:X\rightarrow \{+,-\}$ so that $\ell_i(M(Y,\mathbf{w}),Y_i) > \ell_i(M((Y_{-i},Y'_i),\mathbf{w}),Y_i)$. The mechanism is strongly GSP if there is no subset of agents $S\subseteq [n]$, agent $i^*\in S$, and labelings $Y'_i:X\rightarrow \{+,-\}$ for $i\in S$, so that $\ell_i(M(Y,\mathbf{w}),Y_i)\geq \ell_i(M((Y_{-S},Y'_S),\mathbf{w}),Y_i)$ for each $i\in S$ and $\ell_{i^*}(M(Y,\mathbf{w}),Y_{i^*}) > \ell_{i^*}(M((Y_{-S},Y'_S),\mathbf{w}),Y_{i^*})$. 
%Here, $(Y_{-i},Y'_i)$ (respectively, $(Y_{-S},Y'_S)$) denotes the reported dataset consisting of the private labelings of all agents besides $i$ (respectively, besides the agents in $S$) and labeling $Y'_i$ by agent $i$ (respectively, by every agent $i\in S$). 
A universally SP (resp., strongly GSP) randomized mechanism is a probability distribution over deterministic SP (resp., deterministic strongly GSP) mechanisms.

\citet{MPR12} observe that ERM is not SP and aim to design SP mechanisms which approximate the global risk of ERM. A (possibly randomized) $\rho$-approximate mechanism $M$ in this context satisfies $\E[L(M(Y,\mathbf{w}),Y,\mathbf{w})]\leq \rho\cdot L(\text{erm}(Y,\mathbf{w}),Y,\mathbf{w})$ for every dataset $Y$. They present an $O(n)$-approximate SP mechanism, 
%(essentially returning the classifier of minimum loss for the agent with the highest weight), 
which is also proved to be almost best possible among all deterministic SP mechanisms. A randomized version of this mechanism, which first selects an agent proportionally to the agents' weights and then selects the classifier of minimum loss to the selected agent, is universally SP and has an approximation ratio of $3-\frac{2}{n}$. This bound was proved to be best possible for randomized SP mechanisms by~\citet{MAMR11}.

By exploiting the relation of binary classification with shared inputs to approval voting and our impossibility result, we can show a negative result as well. Consider a set of $m$ points $X\subseteq \mathcal{X}$ and the concept class $\mathcal{H}$ consisting of the classifiers that label $k$ points of $X$ with $+$ and the remaining with $-$, where $k\in [m-2]$. Such classifiers can make sense in scenarios where there is a fixed number of points that can receive one of the two labels, due to capacity constraints. For example, classifiers for school admission or bank loan applications fit under this framework (where only a predetermined number of the applicants will be admitted). Another natural example can be seen with taking $k=n/2$, where the goal is to separate the points into an upper and lower half. 
We can then think of each point as an alternative, classifiers as $k$-sized committees (containing the alternatives that the classifier labels with $+$), and input data provided by the agents as approval ballots, where the approval set corresponds to the alternatives labeled as $+$ by the agent. Notice that under this definition, the loss $\ell_i(h,Y_i)$ of agent $i$ for classifier $h$ is the Hamming distance between the $k$-sized committee corresponding to the classifier $h$ and the approval ballot corresponding to agent $i$. Thus, classification for this particular concept class with agents aiming at minimizing their loss is equivalent to multi-winner approval voting with agents aiming at minimizing their Hamming distance to the winning committee.

We will use now an argument similar to Claim \ref{cl:non-unan} in Section \ref{sec:minimax}, so as to obtain a negative result that extends also for randomized mechanisms. Consider a universally strongly GSP randomized mechanism $M$ which has a strongly GSP deterministic mechanism $M_R$ (corresponding to the multi-winner voting rule $R$) in its support, i.e., it calls mechanism $M_R$ with positive probability. Then, Theorem~\ref{thm:main} implies that $R$ is not unanimous, and hence, there is an approval ballot profile in which all agents agree on $k$ alternatives, but the rule $R$ returns another $k$-sized committee. Equivalently, this means that there exists a dataset $Y=(Y_1, ..., Y_1)$ consisting of $n$ copies of the same labeling $Y_1$ by all agents, so that when the classification mechanism $M_R$ takes as input $Y$, it returns a classifier $h$, producing a labeling that is different from $Y_1$, i.e., $\ell_i(M_R(h,Y,\mathbf{w}),Y_i)>0$ and $L(M_R(Y,\mathbf{w}),Y,\mathbf{w})=\sum_{i\in [n]}{w_i\cdot \ell_i(M_R(Y,\mathbf{w}),Y_1)}>0$ and hence, $\E[L(M(Y,\mathbf{w}),Y,\mathbf{w})]>0$. In contrast, the concept class $\mathcal{H}$ contains the classifier $h^*$ which agrees with the labeling $Y_1$ on the points of $X$ and thus, $L(\text{erm}(Y,\mathbf{w}),Y,\mathbf{w})=0$. Therefore, we obtain that mechanism $M$ has infinite approximation ratio.

\begin{theorem}\label{thm:classification}
    Any universally strongly GSP randomized mechanism for binary classification with shared inputs has an infinite approximation ratio.
\end{theorem}

We remark that we have silently assumed that the classification mechanism may take as input datasets that are not necessarily {\em realizable}. This means that the labeling of the points in $X$ provided by some agent may not coincide with the labeling of any classifier from the concept class. This is crucial for proving Theorem~\ref{thm:classification}. If the labeling provided by each agent is realizable by a classifier in the concept class, then the two mechanisms presented by~\citet{MPR12} are strongly GSP and universally strongly GSP, respectively. A realizable dataset corresponds to an approval voting setting in which the outcome can be {\em any} committee of alternatives. Clearly, any dictatorship is strongly GSP in this setting. For realizable datasets, the mechanisms of~\citet{MPR12} correspond to such dictatorships.

\subsection{Constrained facility location in networks}
Facility location has played a key role in the field of approximate mechanism design with money~\cite{PT13}. The simplest version of the problem aims to locate a single facility on the line when strategic agents report their private location and have {\em single-peaked} preferences. These results usually exploit a famous result by~\citet{M80}, which characterizes the class of SP mechanisms. The results on the line have been extended to multiple dimensions~\cite{BGS93}, multiple facilities~\cite{LSWZ10}, facility location in networks \cite{SV02,AFPT10}, and have addressed constrained versions of the problem; e.g., see~\citet{KVZ19}, \citet{SB15}, and \citet{W21}. 

We focus on the constrained version as well, regarding the allowed locations for placing a facility, and explain how our main result implies an impossibility in discrete metric spaces. In particular, the problem we consider is how to locate a facility at a node of a graph, taking into account reports by agents for their private locations in the graph. Our objective is to locate a facility at a node selected from a predefined subset of allowable nodes so that the maximum or the total shortest-path distance of the facility from all agents is minimized. A strongly GSP mechanism here means that if misreporting by any coalition of agents results in an agent from the coalition coming closer to the facility, some other agent in the coalition will be strictly further off. We evaluate the outcome of such a mechanism in terms of the maximum or the total shortest-path distance compared against the optimal solution. We have the following negative result.

\begin{theorem}\label{thm:facility}
    Strongly GSP mechanisms for constrained facility location in graphs has infinite approximation ratio with respect to both the maximum and the total cost objectives.
\end{theorem}

\begin{proof}
Again, we resort to an argument analogous to Claim \ref{cl:non-unan}. Consider instances, where the graph is the $m$-dimensional hypercube and all nodes with $k$ 1s ($k\in [m-2]$) in their binary representation is the allowable set of nodes for placing the facility. Then, the $m$ dimensions correspond to alternatives, each node in the allowable set corresponds to a $k$-sized committee, and each agent's location in the graph corresponds to an approval ballot. Clearly, the shortest-path distance between two nodes in the hypercube is equal to the Hamming distance of the corresponding approval ballots to committees. Theorem~\ref{thm:main} then implies that every strongly GSP mechanism for locating the facility is non-unanimous. Hence, there exists an instance in which all agents have the same preference, i.e., they are all located in the same node of the allowable set, but the mechanism returns another node of the allowable set as an outcome. Hence, the optimal solution has zero cost, both for the maximum and for the total cost objective, while the mechanism returns a solution with positive cost. 
\end{proof}

Clearly, it is again the restriction of the allowable set of nodes that leads to the impossibility. When the allowable set contains all graph nodes, dictatorships are strongly GSP and have a finite approximation ratio.

\section{Discussion}
Let us conclude by mentioning very briefly the few cases that escape from our impossibility. First, notice that our construction has at least six agents. On the other hand, we can easily see that the serial dictatorship mechanism ---which returns a $k$-sized committee that has minimum distance from agent $2$ among the $k$-sized committees of minimum distance from agent $1$--- is strongly GSP for approval elections with two agents. The case of three, four, and five agents is left open. Furthermore, our construction creates profiles with $m\geq 3$ alternatives and uses committee size $k\in [m-2]$. Also, recall that, for two alternatives, \citet{M12} provides a positive result with his consensus voting rules. So, the case $m\geq 3$ and $k=m-1$ is left mysteriously open. Exploring whether our impossibility carries over even for the simplest among these cases with three alternatives and committees of size $2$ would require a different construction. Of course, the possibility of a positive result here should not be excluded.

%\nocite*
\bibliographystyle{ACM-Reference-Format}
\bibliography{sample}

\end{document}